\documentclass[letterpaper]{llncs}

\usepackage{amsmath,amssymb,amsfonts}
\usepackage{tgpagella,eulervm}
\usepackage{varioref}
\usepackage{enumerate}

\newcommand{\es}{\varnothing}

\DeclareMathOperator{\NP}{\mathsf{NP}}

\begin{document}

\title{{\sc Convex Independence in Permutation Graphs}}

\author{
Wing-Kai~Hon\inst{1}
\and 
Ton~Kloks\inst{}
\and 
Fu-Hong~Liu\inst{1}
\and 
Hsiang-Hsuan~Liu\inst{1,2}
}

\institute{National Tsing Hua University, Hsinchu, Taiwan\\
{\tt (wkhon,fhliu,hhliu)@cs.nthu.edu.tw}
\and
University of Liverpool, Liverpool, United Kingdom\\
{\tt hhliu@liverpool.ac.uk}
}

\maketitle 

\begin{abstract}
A set $C$ of vertices of a graph is $P_3$-convex if every vertex outside  
$C$ has at most one neighbor in $C$. The convex hull $\sigma(A)$ of a set $A$ 
is the smallest $P_3$-convex set that contains $A$. A set $M$ is convexly 
independent if for every vertex $x \in M$, $x \notin \sigma(M-x)$. 
We show that the maximal number of vertices that a convexly independent set 
in a permutation graph can have, can be computed in polynomial time. 
\end{abstract}

\section{Introduction}

Popular models for the spread of disease and of opinion are  
graph convexities. The $P_3$-convexity is one such convexity, and it is 
defined as follows. 

\begin{definition}
A set $S$ of vertices in a graph $G$ is $P_3$-convex if every vertex 
outside $S$ has at most one neighbor in $S$. 
\end{definition}
The $P_3$-convexity will be the only convexity studied in this paper, 
so from now on we use the term convex, instead of $P_3$-convex. 
For a set $A$ of vertices we let $\sigma(A)$ denote its \underline{convex hull}, 
that is, the smallest convex set that contains $A$.\footnote{In his classic paper, 
Duchet defines a 
graph convexity as a collection of `convex' subsets of a (finite) set 
$V$ that contains $\es$ and 
$V$, and that is closed under intersections, and that, furthermore, has the property 
that each convex subset induces a connected subgraph. This last 
condition is, here, omitted.} 

\bigskip 

For a set of points $A$ in $\mathbb{R}^d$ and a point $x$ in its Euclidean 
convex hull, 
there exists a set $F \subseteq A$ of at most $d+1$ points such that $x \in \sigma(F)$, 
ie, $x$ is in 
the Euclidean convex hull of $F$. This is Carath\'eodory's theorem. For convexities in graphs 
one defines the Carath\'eodory number as the smallest number $k$ such that, 
for any set $A$ of vertices, and any vertex $x \in \sigma(A)$, there exists a 
set $F \subseteq A$ with $|F| \leq k$ and $x \in \sigma(F)$. 
For a set $S$, let 
\begin{equation}
\partial(S)= \sigma(S) \setminus \bigcup_{x \in S} \; \sigma(S-x).
\end{equation}
A set is irredundant if $\partial(S) \neq \es$.  Duchet  
showed that the Carath\'eodory 
number is the maximal cardinality of an irredundant set. 

\bigskip 

\begin{definition}
A set $S$ is convexly independent if 
\begin{equation}
\text{\rm for all $x \in S$,} \quad x \notin \sigma(S-x).
\end{equation}
\end{definition}
Notice that, if a set is convexly independent then so is every subset of it 
(since $\sigma$ is a closure operator). 

It appears that there is no universal notation for the maximal 
cardinality of a convexly independent set.\footnote{Ramos et al call it the `rank' 
of the graph, 
but this word has been used for so many different concepts 
that it has lost all meaning.} 
In this paper we denote it 
by $\beta_c(G)$. 
Every irredundant set is convexly independent, thus 
the convex-independence number 
$\beta_c(G)$ is an upperbound for the Carath\'eodory number. 
For example, for paths $P_n$ with $n$ vertices, and for 
cycles $C_n$ with $n$ vertices, we have equality; 
\begin{equation}
\label{eqnpath}
\beta_c(P_n)= 2 \cdot \left\lfloor \frac{n}{3} \right\rfloor + (n \bmod 3) 
\quad\text{and}\quad \beta_c(C_n)=\beta_c(P_{n-1}).
\end{equation}
Other examples, for which the Carath\'eodory number 
equals the convex independence number, 
are \underline{leafy trees}, which are trees with at most one vertex of degree two. 
It is easy to check, that 
\begin{equation}
\label{eqntree}
\text{$T$ is a leafy tree}\quad \Rightarrow\quad
\beta_c(T)=\text{the number of leaves in $T$.}
\end{equation}
Examples for which the Carath\'eodory number is strictly less  
than the convex-independence number are disconnected graphs. 
If $S$ is an irredundant set 
then $\sigma(S)$ is necessarily connected. However, 
$\beta_c(G)$ is the sum of the convexly independence 
numbers of $G$'s components. Notice also that $\beta_c(P_6)=4$, but 
there exists a maximum convexly independent set $S$ 
for which $\sigma(S)=S$ and is disconnected, and, thence, redundant. 

\bigskip 

A set $S$ is a 2-packing if it is an independent set in $G^2$, that is, 
no two vertices of $S$ are adjacent or have a common neighbor. 
Every 2-packing $S$  
is convexly independent, as $\sigma(S)=S$. 
For splitgraphs with minimal degree at least two, 
a maximal convex-independent set is a 2-packing, unless 
it has only two vertices. It follows that computing 
the convexly 
independence number 
is $\NP$-complete for splitgraphs (it is
Karp's {\sc set packing}, problem~4).
For biconnected chordal 
graphs (including the splitgraphs mentioned above), 
every vertex is in the convex hull of any set of 
two vertices at distance at most two. Thus, the Carath\'eodory number for those is 
two. 

\bigskip 

Ramos et al show that computing the convexly independence number remains 
$\NP$-complete for bipartite graphs, and they show that it is polynomial 
for trees and for threshold graphs. 

\bigskip 

The intersection graph of a collection of straight line segments, with endpoints 
on two parallel (horizontal) lines, is called a permutation graph. 
Dushnik and Miller 
characterize them as the comparability 
graphs for which the complement is a comparability graph as well.  
In this paper we show that the convexly independence number of permutation 
graphs is computable in polynomial time. 

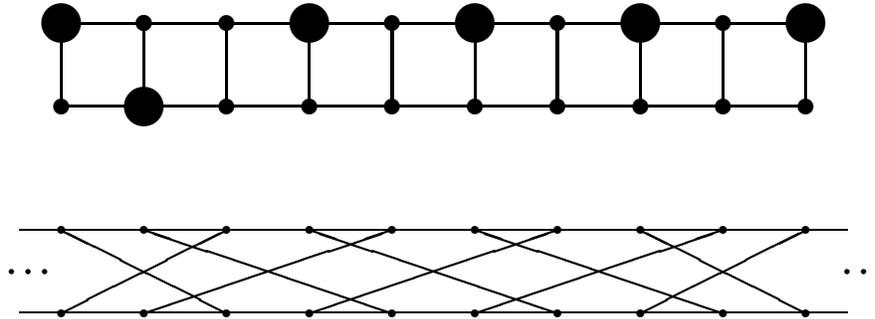
\begin{figure}
\begin{center}
\setlength{\unitlength}{1.1mm}
\thicklines
\begin{picture}(110,50)
\put(10,35){\circle*{2.0}}
\put(20,45){\circle*{2.0}}
\put(20,35){\circle*{2.0}}
\put(20,45){\circle*{2.0}}
\put(30,35){\circle*{2.0}}
\put(30,45){\circle*{2.0}}
\put(40,35){\circle*{2.0}}
\put(40,45){\circle*{2.0}}
\put(50,35){\circle*{2.0}}
\put(50,45){\circle*{2.0}}
\put(60,35){\circle*{2.0}}
\put(60,45){\circle*{2.0}}
\put(70,35){\circle*{2.0}}
\put(70,45){\circle*{2.0}}
\put(80,35){\circle*{2.0}}
\put(80,45){\circle*{2.0}}
\put(90,35){\circle*{2.0}}
\put(90,45){\circle*{2.0}}
\put(100,35){\circle*{2.0}}
\put(100,45){\circle*{2.0}}
\put(10,35){\line(1,0){90}}
\put(10,45){\line(1,0){90}}
\put(10,35){\line(0,1){10}}
\put(20,35){\line(0,1){10}}
\put(30,35){\line(0,1){10}}
\put(40,35){\line(0,1){10}}
\put(50,35){\line(0,1){10}}
\put(60,35){\line(0,1){10}}
\put(70,35){\line(0,1){10}}
\put(80,35){\line(0,1){10}}
\put(90,35){\line(0,1){10}}
\put(100,35){\line(0,1){10}}
\put(20,35){\circle*{5.0}}
\put(10,45){\circle*{5.0}}
\put(40,45){\circle*{5.0}}
\put(60,45){\circle*{5.0}}
\put(80,45){\circle*{5.0}}
\put(100,45){\circle*{5.0}}

\thinlines
\put(5,10){\line(1,0){100}}
\put(5,20){\line(1,0){100}}
\thicklines
\put(10,10){\line(2,1){20}}
\put(20,10){\line(3,1){30}}
\put(40,10){\line(3,1){30}}
\put(60,10){\line(3,1){30}}
\put(80,10){\line(2,1){20}}
\put(30,10){\line(-2,1){20}}
\put(50,10){\line(-3,1){30}}
\put(70,10){\line(-3,1){30}}
\put(90,10){\line(-3,1){30}}
\put(100,10){\line(-2,1){20}}
\put(10,10){\circle*{1}}
\put(10,20){\circle*{1}}
\put(20,10){\circle*{1}}
\put(30,10){\circle*{1}}
\put(40,10){\circle*{1}}
\put(50,10){\circle*{1}}
\put(60,10){\circle*{1}}
\put(70,10){\circle*{1}}
\put(80,10){\circle*{1}}
\put(90,10){\circle*{1}}
\put(100,10){\circle*{1}}
\put(20,20){\circle*{1}}
\put(30,20){\circle*{1}}
\put(40,20){\circle*{1}}
\put(50,20){\circle*{1}}
\put(60,20){\circle*{1}}
\put(70,20){\circle*{1}}
\put(80,20){\circle*{1}}
\put(90,20){\circle*{1}}
\put(100,20){\circle*{1}}

\multiput(8,15)(-2,0){3}{\circle*{.6}}
\multiput(105,15)(2,0){3}{\circle*{.6}}

\end{picture}
\end{center}
\caption{The heavy dots specify an irredundant set $S$ with $\sigma(S)=V$. The second 
figure shows the intersection model, ie, the `permutation diagram,' for ladders. 
This example shows that the Carath\'eodory number for biconnected 
permutation graphs is unbounded.} 
\end{figure}

\bigskip 

This seems a good time to do a warm-up; 
let's have a close look at convex independence in cographs. 

\section{Convex independence in cographs}

\begin{definition}
A graph is a cograph if it has no induced $P_4$, the path with 4 vertices. 
\end{definition}

Cographs are characterized by the property that every induced subgraph 
is disconnected or else, its complement is disconnected. In other words, cographs 
allow a complete decomposition by joins and unions. 
It follows that cographs are permutation 
graphs, as also this class is closed under joins and unions. 

\bigskip 

Ramos et al analyze the 
convex-independence number for 
threshold graphs. Threshold graphs are the graphs without induced $P_4$, $C_4$ and $2K_2$, 
hence, threshold graphs are properly contained in the class of cographs. In the following 
theorem we extend their results. 

\begin{theorem}
There exists a linear-time algorithm to compute the convex-independence 
number of cographs.
\end{theorem}
\begin{proof}
Let $G$ be a cograph. First assume that $G$ is a union of two smaller 
cographs, $G_1$ and 
$G_2$. In that case, the convex-independence number of $G$ is the sum of 
$\beta_c(G_1)$ and $\beta_c(G_2)$, that is, 
\begin{equation}
G=G_1 \oplus G_2 \quad\Rightarrow \quad 
\beta_c(G)=\beta_c(G_1) + \beta_c(G_2).
\end{equation}

\medskip 

\noindent 
Now, assume that $G$ is a join of two smaller cographs $G_1$ and $G_2$. 
In that case, every vertex of $G_1$ is adjacent to every vertex of $G_2$. 
Let $S$ be a convex-independent set. If $S$ has at least one vertex in 
$G_1$ and at least one vertex in $G_2$, then $|S|=2$, since $G[S]$ cannot 
have an induced $P_3$ or $K_3$. 

\medskip 

\noindent 
Consider a convex-independent set $S \subseteq V(G_1)$.  
Assume that $|S|>1$ and that $|V(G_2)| \geq 2$. Then,
any two vertices  
of $S \cap V(G_1)$ generate $V(G_2) \subseteq \sigma(S)$, and, in turn, 
$V(G)$ is in their convex hull. 
This implies that $S$ 
cannot have any other vertices, that is, 
\[|V(G_2)| \geq 2 \quad\Rightarrow\quad |S| \leq 2.\] 

\medskip 

\noindent
Next, assume 
\[\boxed{S \subseteq V(G_1) \quad\text{and}\quad |V(G_2)|=1.}\] 
Say $u$ is in the singleton $V(G_2)$, that is, $u$ is a universal vertex.  
Let $C_1,\dots,C_t$ be the 
components of $G_1$.  We claim that
\[ |S \cap C_i| \leq \min\;\{\;2,\;|C_i|\;\}.\] 
To see that, assume that $|C_i| \geq 2$. Then, since $G[C_i]$ is a 
connected cograph, $G[C_i]$ is the 
join of two cographs, say with vertex sets $A$ and $B$. 
If $S$ has three vertices in $A$, then each of them is in the 
convex hull of the other two, since $B \cup \{u\}$ is contained 
in their common 
neighborhood, and this set contains at least two vertices. 

\medskip 

\noindent
Assume $S$ has vertices in at least two 
different components of $G_1$.    
Assume furthermore that one component $C_i$ has at least two 
vertices of $S$, say $p$ and $q$. Let $\zeta$ be a vertex of $S$ in 
another component. Then $u \in \sigma(\{p,\zeta\})$, because $[p,u,\zeta]$ is 
an induced $P_3$. 

\medskip 

\noindent
The induced subgraph $G[C_i]$ is a join of two smaller cographs, 
say with vertex sets $A$ and $B$.  
If $p$ and $q$ are both in $A$, 
then $q \in \sigma(S-q)$, since $p$ and $u$ generate $B \subset \sigma(S)$, 
and $B \cup \{u\}$  contains two neighbors of $q$. 
If $p \in A$ and $q \in B$, then $q$ has two neighbors in $\sigma(S-q)$, 
namely $p$ and $u$. Thus, again, $q \in \sigma(S-q)$. 

\medskip 

\noindent
In fine, either each component of $G_1$ contains one vertex of $S$, or else
$|S| \leq 2$. 

\medskip 

\noindent 
This proves the theorem. 
\qed\end{proof}

\section{Monadic second-order logic}

In this section we show that the maximal cardinality of a 
convex-independent set is computable in linear time for graphs 
of bounded treewidth or rankwidth. 
To do that, we show that there is a formulation of the problem 
in monadic second-order logic. The claim then follows from Courcelle's 
theorem. 

\bigskip 

By definition, a set of vertices $W \subseteq V$ is convex if 
\begin{equation}
\label{eq1}
\forall_{x \in V} \; x \notin W \quad\Rightarrow\quad |N(x) \cap W| \leq 1.
\end{equation}

\bigskip 

Let $S \subseteq V$. 
To formulate that a set $W = \sigma(S)$ we formulate that 
\begin{enumerate}[\rm 1.]
\item $S \subseteq W$, and 
\item $W$ satisfies~\eqref{eq1}, and 
\item For all $W^{\prime}$ for which the previous two conditions hold, 
$W \subseteq W^{\prime}$.
\end{enumerate}

\bigskip 

Finally, a set $S$ is convexly independent if 
\begin{equation}
\label{eq2}
\forall_{x \in V} \; x \in S \quad\Rightarrow \quad x \notin \sigma(S-x).
\end{equation}
Actually, 
to show that $x \notin \sigma(S-x)$ it is sufficient to formulate that 
(for every vertex $x \in S$) there 
is a set $W_x$ such that 
\begin{equation}
\label{eq3}
\text{$W_x$ is convex}\quad\text{and}\quad S\setminus \{x\} \subseteq W_x 
\quad\text{and}\quad x \notin W_x.
\end{equation}

\bigskip 

The formulas \eqref{eq1}---\eqref{eq3} 
show that convex independence can be formulated in 
monadic second-order logic (without quantification over subsets of edges). 
By Courcelle's theorem we obtain the following. 
\begin{theorem}
\label{thm bounded tw}
There exists a linear-time algorithm to compute the convex-independence number 
for graphs of bounded treewidth or rankwidth. 
\end{theorem}

\begin{remark}
Notice that also the Carath\'eodory number is expressible in monadic 
second-order logic.
\end{remark}

\subsection{Trees}

Let $T$ be a tree with $n$ vertices and maximal degree $\Delta$. Ramos et al 
present an involved 
algorithm, that runs in $O(n \log \Delta)$ time, to compute a 
convexly independent set. By Theorem~\ref{thm bounded tw}, 
there exists a linear-time algorithm 
that accomplishes this.
We propose a different algorithm. 

\begin{theorem}
There exists a linear-time algorithm that computes the 
convexly independence number of trees. 
\end{theorem}
\begin{proof}
Let $T$ be a tree. Decompose $T$ into a minimal number of 
maximal, vertex-disjoint, leafy trees, $F_1, \dots, F_s$, 
and a collection of paths.  The endpoints of the paths are 
separate leaves of the trees $F_i$ or pendant vertices. 
By Equation~\eqref{eqntree}, 
each leafy tree $F_i$ has a maximum convexly independent set 
consisting of its leaves.  The convexly independence numbers  
of the connecting paths are given by Equation~\eqref{eqnpath}. 
Notice that each path has a maximum, convexly independent set 
that contains the two endpoints. 
\qed\end{proof}

\section{The convex-independence number of permutation graphs}

In this section we show that there exists a polynomial-time algorithm 
to compute the convex-independence number of permutation graphs. 

\bigskip 

In the following discussion, let $G$ be a permutation graph with a 
fixed permutation diagram. 
We refer to $S$ as a generic convex-independent set in $G$. 

\begin{definition}
Let $S$ be a convex-independent set in $G$. Let $x \in V \setminus S$. 
A \underline{2-path} connecting $x$ to $S$ is a sequence of vertices 
\begin{equation}
\Delta=[s_1,s_2,x_1,\dots,x]
\end{equation}
in which every vertex has two neighbors that appear earlier in the sequence, 
or else it is in $S$. 
\end{definition}

\begin{lemma}
\begin{equation}
x \in \sigma(S) \quad \Leftrightarrow \quad 
\text{\rm there is a 2-path connecting $x$ to $S$.}
\end{equation}
\end{lemma}
\begin{proof}
Following Duchet, let $I(x,y)$ be the `interval function' of 
the $P_3$-convexity, that is, for two vertices 
$x$ and $y$, $I(x,y)$ is $\{x,y\}$ plus 
the set of vertices that are adjacent to 
both $x$ and $y$.\footnote{Duchet proved that, for interval convexities,  
the Carath\'eodory number is the smallest integer $k \in \mathbb{N}$ such that every 
$(k+1)$-set is redundant. Thus, the fixed-parameter Carath\'eodory 
number is polynomial.} 
For a set $S$, we let 
\begin{equation}
I^0(S)=S \quad\text{and}\quad I^{k+1}(S)=I(I^k(S) \times I^k(S)).
\end{equation}
Then, 
\begin{equation}
\sigma(S)=\bigcup_{k \in \mathbb{N} \cup \{0\}} \; I^k(S).
\end{equation}
In other words, a vertex is in $I^{k+1}(S)$ if 
it is in $I^k(S)$, or else it has two neighbors in $I^k(S)$. This is expressed 
by the existence of a 2-path. 

\medskip 

\noindent
This proves the lemma.
\qed\end{proof}

\bigskip 

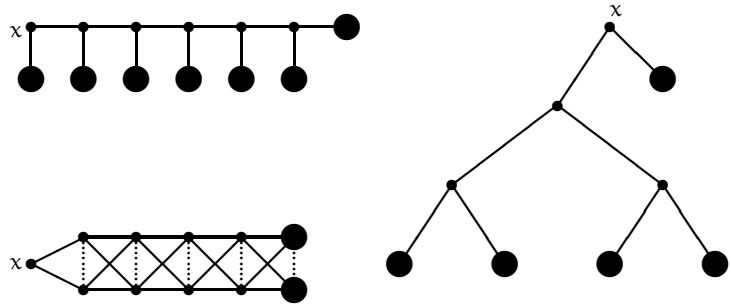
\begin{figure}
\begin{center}
\setlength{\unitlength}{.7mm}
\thicklines
\begin{picture}(130,70)
\put(10,15){\circle*{2}}
\put(6,14){$x$}
\put(20,10){\circle*{2}}
\put(20,20){\circle*{2}}
\put(30,10){\circle*{2}}
\put(30,20){\circle*{2}}
\put(40,10){\circle*{2}}
\put(40,20){\circle*{2}}
\put(50,10){\circle*{2}}
\put(50,20){\circle*{2}}
\put(60,10){\circle*{5}}
\multiput(60,10)(0,1){10}{\circle*{.5}}
\put(60,20){\circle*{5}}
\put(10,15){\line(2,-1){10}}
\put(10,15){\line(2,1){10}}
\multiput(20,10)(0,1){10}{\circle*{.5}}
\put(20,10){\line(1,1){10}}
\put(20,10){\line(1,0){10}}
\put(20,20){\line(1,-1){10}}
\put(20,20){\line(1,0){10}}
\put(30,10){\line(1,0){10}}
\multiput(30,10)(0,1){10}{\circle*{.5}}
\put(30,10){\line(1,1){10}}
\put(30,20){\line(1,0){10}}
\put(30,20){\line(1,-1){10}}
\put(40,10){\line(1,0){10}}
\multiput(40,10)(0,1){10}{\circle*{.5}}
\put(40,10){\line(1,1){10}}
\put(40,20){\line(1,0){10}}
\put(40,20){\line(1,-1){10}}
\put(50,10){\line(1,0){10}}
\multiput(50,10)(0,1){10}{\circle*{.5}}
\put(50,10){\line(1,1){10}}
\put(50,20){\line(1,0){10}}
\put(50,20){\line(1,-1){10}}

\put(10,50){\circle*{5}}
\put(10,60){\circle*{2}}
\put(6,58){$x$}
\put(20,50){\circle*{5}}
\put(20,60){\circle*{2}}
\put(30,50){\circle*{5}}
\put(30,60){\circle*{2}}
\put(40,50){\circle*{5}}
\put(40,60){\circle*{2}}
\put(50,50){\circle*{5}}
\put(50,60){\circle*{2}}
\put(60,50){\circle*{5}}
\put(60,60){\circle*{2}}
\put(70,60){\circle*{5}}
\put(10,60){\line(1,0){60}}
\put(10,60){\line(0,-1){10}}
\put(20,60){\line(0,-1){10}}
\put(30,60){\line(0,-1){10}}
\put(40,60){\line(0,-1){10}}
\put(50,60){\line(0,-1){10}}
\put(60,60){\line(0,-1){10}}

\put(120,60){\circle*{2}}
\put(130,50){\circle*{5}}
\put(120,60){\line(-2,-3){10}}
\put(120,60){\line(1,-1){10}}
\put(120,62){$x$}
\put(80,15){\circle*{5}}
\put(90,30){\circle*{2}}
\put(100,15){\circle*{5}}
\put(110,45){\circle*{2}}
\put(120,15){\circle*{5}}
\put(130,30){\circle*{2}}
\put(140,15){\circle*{5}}
\put(110,45){\line(-4,-3){20}}
\put(110,45){\line(4,-3){20}}
\put(90,30){\line(-2,-3){10}}
\put(90,30){\line(2,-3){10}}
\put(130,30){\line(-2,-3){10}}
\put(130,30){\line(2,-3){10}}

\end{picture}
\end{center}
\caption{The figure shows three examples of 2-paths from a vertex $x$ to a set $S$. 
The heavy dots 
represent vertices of $S$. Notice, however, that the binary tree is not a permutation 
graph (since it has an asteroidal triple). The second example is a simple path in 
which each vertex, except $x$, is replaced by a twin. Permutation graphs are closed 
under creating twins, so, since paths 
are permutation graphs, this second example is so also.}
\end{figure}

\begin{lemma}
\label{lm linear order}
Each component of $G[S]$, ie, the subgraph induced by $S$, 
is a single vertex or an edge. 
In a permutation diagram for $G$, there is a linear left-to-right ordering 
of the components of $G[S]$. 
\end{lemma}
\begin{proof}
Since $S$ is convexly independent, $G[S]$ cannot contain $K_3$ or $P_3$. 
Thus each component of $G[S]$ is an edge or a vertex. 

\medskip 

\noindent 
Fix a permutation diagram for $G$. The line segments that correspond to 
the vertices of $S$ form a permutation diagram for $G[S]$. Each component 
of $G[S]$ is a connected part of the diagram, and the left-to-right 
ordering of the connected parts in the diagram yields a total 
ordering of the components of $G[S]$. 
\qed\end{proof}

\begin{definition}
The last component of $S$ is the rightmost component in the linear 
ordering as specified in Lemma~\ref{lm linear order}. 
\end{definition}

We say that a vertex $\notin S$ is to the right of the last component if 
its line segment appears to the right of the last component, ie, the endpoints 
of the line segment, on the top line and 
bottom line of the diagram, appear to the right of the endpoints of the last component. 

\begin{definition}
The \underline{border} of 
$\sigma(S)$ is the set of the two rightmost endpoints, on the top line 
and bottom line of the permutation diagram, that are endpoints of line segments 
corresponding to vertices of $\sigma(S)$. 
\end{definition}

We say that a line segment is to the left of the border if both its endpoints 
are left of the appertaining endpoints that constitute the border. 

\begin{lemma}
\label{lm border in S}
If the elements of the border of $\sigma(S)$ are the endpoints of a single line 
segment, then this is the line segment of a vertex in $S$. 
\end{lemma}
\begin{proof}
Let $x$ be the vertex whose line segment has endpoints that form the border of 
$\sigma(S)$. Assume that $x \notin S$. Then, by definition, there is a 2-path $\Delta$ 
from $x$ to $S$ and all vertices of the 2-path are in $\sigma(S)$. Since $x \notin S$, 
it has two neighbors that appear earlier in $\Delta$. The line segments of the 
two neighbors are crossing the line segment of $x$, and so there must be 
endpoints of $\sigma(S)$ that appear to the right of the endpoints of $x$. This is a 
contradiction. 
\qed\end{proof}

\begin{lemma}
\label{lm vertex disjoint paths}
For every vertex in $\sigma(S) \setminus S$ there exist two vertex-disjoint 
paths to $S$ with all vertices in $\sigma(S)$. 
\end{lemma}
\begin{proof}
We may assume that $|S| > 1$, otherwise $\sigma(S)=S$ and the claim is void. 
For convenience, add edges to the graph such that $S$ becomes a clique and 
remove the 
vertices that are not in $\sigma(S)$. 
Assume that some vertex of $\sigma(S) \setminus S$ 
is separated from $S$ by a cutvertex $c$. 
Since $S$ is a clique, $S-c$ is contained in one component $C_1$ of $\sigma(S)-c$ 
and some vertices of $\sigma(S) \setminus S$ are in some other component $C_2$. 
Consider all 2-paths from vertices in $C_2$ to $S$.  Let $x$ be a vertex 
that is in $C_2$ with a shortest 2-path to $S$.  
Then $x$ must have two neighbors that appear earlier in the 2-path. 
But that is impossible, 
since there is only one candidate, namely $c$. 
\qed\end{proof}

\begin{lemma}
\label{lm decide sigma}
Let the line segment of a vertex $x$ be to the right of the last component 
of $S$. Then $x \in \sigma(S)$ if and only if $x$'s line segment 
is to the left of the border. 
\end{lemma}
\begin{proof}
By Lemma~\ref{lm border in S}, 
we may assume that the border corresponds to two adjacent 
vertices $a$ and $b$. By definition of the border, $x \in \sigma(S)$ 
implies that both of $x$'s endpoints are left of the border elements. 
If $x$ is adjacent to both $a$ and $b$ then $x$ has two neighbors in 
$\sigma(S)$, and so, $x$ itself is in $\sigma(S)$. 
Assume that $x \notin N(a)$. Since $a \in \sigma(S)\setminus S$, 
by Lemma~\ref{lm vertex disjoint paths}, there are two vertex-disjoint 
paths from $a$ to $S$. 
Since the line segment of $x$ is between $a$ and the last component of $S$, 
each path must contain a neighbor of $x$. This implies that $x \in \sigma(S)$. 
\qed\end{proof}

\bigskip 

Our algorithm performs a dynamic programming on feasible last components 
of $S$ and the border of $\sigma(S)$.  
By Lemma~\ref{lm decide sigma}, the last element of $S$ and the border 
supply sufficient information to decide whether a `new last component,'  to the 
right of the previous last component, has a vertex in $\sigma(S)$ or not.   

\bigskip 

Let $S^{\ast} = S \cup X$, where $X$ is either a single vertex or an edge, and 
assume that $X$ has no vertex $x \in X \cap \sigma(S^{\ast}-x)$. 

\begin{remark}
Notice that, if $|X|=1$, then it may be adjacent to one vertex appertaining the border. 
(For an example, see the figure below.) 
When $|X|=2$, both vertices of $X$ must be to the right of the border, otherwise, 
one element of $X$ is adjacent to an element of $\sigma(S)$ and to 
the other element 
of $X$, which would make $S^{\ast}$ convexly dependent. In any case, given the border, 
it is easy to check, algorithmically, 
the feasibility of a new component $X$. 
\end{remark}

To guarantee that $S^{\ast}$ is a 
convex-independent set (with last component $X$), we need to check that, for any 
$s \in S$, $s \notin \sigma(S^{\ast}-s)$.  The figure below shows that the last 
component of $S$ and the border of $\sigma(S)$ do not convey sufficient 
information to guarantee that $S^{\ast}$ is convexly independent. 

\begin{figure}
\begin{center}
\setlength{\unitlength}{1.1mm}
\begin{picture}(110,40)
\put(10,10){\circle*{1}}
\put(20,10){\circle*{1}}
\put(30,10){\circle*{1}}
\put(40,10){\circle*{1}}
\put(50,10){\circle*{1}}
\put(10,30){\circle*{1}}
\put(20,30){\circle*{1}}
\put(30,30){\circle*{1}}
\put(40,30){\circle*{1}}
\put(50,30){\circle*{1}}
\put(5,10){\line(1,0){50}}
\put(5,30){\line(1,0){50}}
\thicklines
\put(10,10){\line(1,2){10}}
\put(10,30){\line(1,-1){20}}
\put(30,30){\line(1,-2){10}}
\put(40,30){\line(1,-2){10}}
\thinlines
\put(20,10){\line(3,2){30}}

\put(9,8){$s_2$}
\put(19,8){$y$}
\put(29,8){$s$}
\put(39,8){$s_1$}
\put(49,8){$x$}
\put(9,31){$s$}
\put(19,31){$s_2$}
\put(29,31){$s_1$}
\put(39,31){$x$}
\put(49,31){$y$}

\put(80,15){\circle*{2}}
\put(80,25){\circle*{2}}
\put(90,15){\circle*{2}}
\put(90,25){\circle*{2}}
\put(100,15){\circle*{2}}
\put(80,15){\line(1,0){20}}
\put(80,15){\line(0,1){10}}
\put(90,15){\line(0,1){10}}
\put(77,14){$s$}
\put(89,12){$y$}
\put(102,14){$x$}
\put(79,27){$s_2$}
\put(89,27){$s_1$}

\end{picture}
\end{center}
\caption{The figure shows $s\in\sigma(S^{\ast}-s)$, where $S=\{s,s_1,s_2\}$,  
$S^{\ast}=\{s,s_1,s_2,x\}$, $\sigma(S)=\{s,s_1,s_2,y\}$ and 
$\sigma(S^{\ast}-s)=\{s_1,s_2,x,y_2,s\}$. Notice that $x \notin \sigma(S)$ ($x$ is 
not left of the border; the border of $\sigma(S)$ has the endpoint of $y$ on the 
top line and the endpoint of $s_1$ on the bottom line).}
\end{figure}
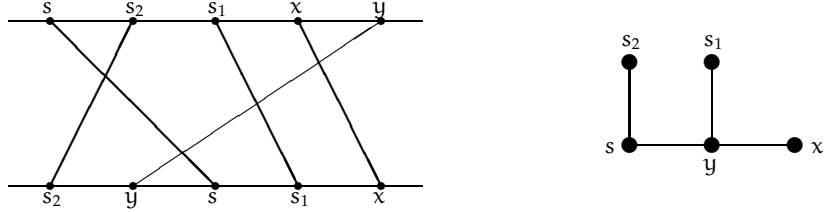

\bigskip 

In the following, let $L$ be the last component of $S$, let $X$ be a 
feasible, `new,' last component of $G[S^{\ast}]$, to the right of $L$, 
where $S^{\ast}=S \cup X$. 
(Of course, both $L$ and $X$ are either single vertices or edges, and no vertex 
of $L$ is adjacent to any vertex of $X$.) 
The feasibility of $X$ is defined so that: 
\begin{equation}
\label{eqnx}
\forall_{x \in X}\; x \notin \sigma(S^{\ast}-x).
\end{equation}
In our final theorem, below, we prove that it is sufficient to maintain 
a constant amount of information, to enable the algorithm 
to check the convexly independence of $S^{\ast}$. 

\bigskip 

First we define a partial 2-path.  We consider two cases, namely where 
$|X|=1$ and where $|X|=2$. To define it, we `simulate' $X$ by an auxiliary 
vertex, or an auxiliary true twin, that we place immediately to the right of $L$.  

\begin{definition}
When $|X|=1$, add one vertex $s^{\prime}$ with a line segment whose endpoints 
are immediately to the right of the rightmost endpoint of $L$ on the top 
line and the rightmost endpoint of $L$ on the bottom line. When $|X|=2$, then 
replace the vertex $s^{\prime}$ above by a true twin $s_1^{\prime}$ and $s_2^{\prime}$. 
Let $X^{\prime}=\{s^{\prime}\}$ when $|X|=1$ and $X^{\prime}=\{s_1^{\prime},s_2^{\prime}\}$ 
when $|X|=2$. Finally, let $S^{\prime}=S \cup X^{\prime}$. 
A partial 2-path from $u \in S$ to $S^{\ast}-u$ is a 2-path from $u$ to 
$S^{\prime}-u$, from which $S^{\prime}$ is removed. 
\end{definition}

\begin{theorem}
There exists a polynomial-time algorithm to compute a convexly independent 
set of maximal cardinality in permutation graphs. 
\end{theorem}
\begin{proof}
Consider a vertex $u \in S$ for which $u \in \sigma(S^{\ast}-u)$. We may assume 
that $u \notin L$, because $L$ is available to the algorithm and so, 
it is easy to check 
the condition for elements of $L$. 
Then there is a 2-path $\Delta=[s_1,s_2,\dots,u]$ from $u$ to $S^{\ast}-u$. 
If no vertex of $X$ is in this 2-path, then $u \in \sigma(S-u)$, which 
contradicts our assumption that $S$ is convexly independent. 
We may assume that at least one of $s_1$ and $s_2$ 
is an element of $X$. 

\medskip 

\noindent
Since $\Delta$ contains two vertex-disjoint paths from $u$ to $S^{\ast}-u$, 
at least one of these paths must contain some vertex of $N(L)$. 
Partition the vertices of $N(L)$ in two parts. One part contains 
those vertices that have their endpoint on the top line to the right of 
$L$, and the other part contains those vertices that have their endpoint on the 
bottom line to the right of $L$. We claim that both parts are totally 
ordered by set-inclusion of 
their neighborhoods in the component of $G-N[L]$ that contains $X$. 
To see that, consider two elements $a$ and $b$ of $N(L)$. Say that $a$ and 
$b$ both have an endpoint on the top line, to the right of $L$. If 
that endpoint of $a$ is to the left of the endpoint of $b$, then every neighbor 
of $a$ in the component of $G-N[L]$ that contains $X$ is also a neighbor of $b$. 

\medskip 

\noindent
We store subsets with two vertices, $y_1$ and $y_2$ in $N(L)$, 
for which there is a 
partial 2-path from some vertex $u \in S$ to $y_1$ and $y_2$. It is sufficient 
to store only those two vertices 
$y_1$ and $y_2$ that they have a maximal neighborhood. 
In other words, we choose 
$y_1$ and $y_2$ such that their endpoints on the top line and 
bottom line are furthest to the right, or, if they are 
both in the same part, the two that have a maximal neighborhood. 

\medskip 

\noindent 
One other possibility is, that a 
2-path from $u$ to $\sigma(S^{\ast}-u)$ has only one vertex $y \in N(L)$ 
on a path from $u$ to $X$.   Of those 2-paths, 
We also store the element $y$, with a largest 
neighborhood in the component of $G-N[L]$ that contains $X$. 

\medskip 

\noindent 
To check if $S^{\ast}$ is convexly independent it is now sufficient to check if 
one of the partial paths to $y_1$ and $y_2$, or to the 
single element $y$, extend to $X$. 

\medskip 

\noindent 
This proves that there is a dynamic programming algorithm to compute 
a maximum convexly independent set. 
\qed\end{proof}

\end{document}